\documentclass{amsart}
\usepackage{graphicx}
\usepackage{dcolumn,bm,graphicx,color,amsmath,amsthm}
\usepackage{csquotes}
\usepackage[all]{xy}
\title[Darboux transformations for 
operators $\partial_x\partial_y +a \partial_x + b\partial_y +c$]{Classification of Darboux transformations for 
operators of the form $\partial_x\partial_y +a \partial_x + b\partial_y +c$}
\author{Ekaterina SHEMYAKOVA}
\address{Department of Mathematics and Statistics, The University of Toledo, Ohio, USA}
\email{Ekaterina.Shemyakova@UToledo.edu}

\newcommand{\dee}{\partial}

\newcommand{\Ko}{K}

\newcommand{\vars}{x,y, \dots}

\newcommand{\K}{\ensuremath K[\partial_x,\partial_y, \dots]}

\def\t#1#2#3{#1\stackrel{#2}{\longrightarrow}{#3}}
\def\o#1{\ensuremath #1}

\newcommand{\comm}[1]{\textcolor{black}{#1}}

\newtheorem{theorem}{Theorem}[section]
\newtheorem{lemma}{Lemma}[section]
\newtheorem{proposition}{Proposition}[section]
\newtheorem{corollary}[theorem]{Corollary}
\theoremstyle{remark}

\theoremstyle{definition}
\newtheorem{definition}{Definition}[section]

\begin{document}
\begin{abstract}
Darboux transformations are non-group type symmetries of linear  differential operators.
One can define Darboux transformations algebraically by the  intertwining relation $ML=L_1M$ or the  intertwining relation 
$ML=L_1N$ in the cases when the former is too restrictive.

Here we show that Darboux transformations for operators of the form $L=\partial_x\partial_y +a \partial_x + b\partial_y +c$
(sometimes referred to  as 2D Schr\"odinger operators or as to Laplace operators) are always   compositions of atomic Darboux transformations of two different well-studied types of Darboux transformations, provided that the chain of Laplace transformations for the original operator is long enough. 
%
%
\end{abstract}



\maketitle

\section{Introduction}

A Darboux transformation (DT), in the general sense of the word, is a non-group symmetry 
of linear  differential operators (partial or ordinary), which simultaneously transforms the  kernels (solution spaces) or eigenspaces.  Darboux transformations form a category. 
Their structure is under active investigation. See, in particular, 
recent (2018) works of Ch.~Athorne~\cite{Athorne2018}, S.~Smirnov~\cite{Smirnov2018},
and G.~Hovhannisyan et al.~\cite{HOVHANNISYAN2018776, HOVHANNISYAN20161690}.

Darboux transformations originated in the work of Darboux and others on the theory of surfaces, as in~\cite{Darboux2}, 
while particular examples were known to  Euler and Laplace. In XXth century, Darboux transformations were rediscovered 
in quantum mechanics~\cite{shroedinger_on_DT,infeld_hull_factorization_method}, and later, in 1970s,
in Integrable Systems.
A large number of works can be mentioned here, 
but the very name `Darboux transformations' was introduced by V.~B.~Matveev in 1979 in a series of works, see e.g.~\cite{Matveev79}, 
and then the theory was elaborated further in the fundamental monograph~\cite{matveev:1991:darboux} by V.~B.~Matveev and M.~A.~Salle (see also~\cite{doktorovleble2007dressing}).


The model example of Darboux transformations is the following transformation of single variable Sturm--Liouville operators: $L\to L_1$, where $L=\dee_x^2+u(x)$,  $L_1=\dee_x^2+u_1(x)$,  and $u_1(x)$ is 
obtained from $u(x)$ by the formula $u_1(x)=u(x) + 2(\ln\varphi_0(x))_{xx}$. Here $\varphi_0(x)$ is a `seed' solution of the Sturm--Liouville equation $L\varphi_0=\lambda_0\varphi_0$ (with some fixed $\lambda_0$). 
Then the transformation $\widetilde{\varphi} = (\dee_x-(\ln\varphi_0)_x)\varphi$ sends solutions of $L\varphi=\lambda\varphi$ to solutions of $L_1\widetilde{\varphi}=\lambda\widetilde{\varphi}$ (with the same $\lambda$). The seed solution is mapped to zero. (According to~\cite{novikov-dynnikov:1997}, this transformation was already known to Euler.) This example is  model in two ways. 

First, the formula for the transformation of solutions can be written in terms of a Wronskian determinant: $\varphi_1=W(\varphi_0,\varphi)/\varphi_0$. This generalizes to a construction based on several linearly independent seed solutions and higher order Wronskian determinants (Crum~\cite{Crum:1955:1999} for Sturm--Liouville operators and Matveev~\cite{Matveev79} for general operators on the line). 

Secondly, if one lets $M=\dee_x-(\ln\varphi_0)_x$, then the following identity is satisfied:
\begin{equation}\label{intertw1}
ML=L_1M\,.
\end{equation}
This identity is equivalent to the relation between the old potential $u(x)$ and the new potential $u_1(x)$.

This example can be generalized to the case of two operators with the same principal symbol, $L$ and $L_1$, 
that satisfy~\eqref{intertw1} for some $M$. In this case,~\eqref{intertw1} is called the \emph{intertwining relation}.
One can see that if~\eqref{intertw1} is satisfied,  then the operator $M$ defines a linear transformation of the eigenspaces of $L$ to the  eigenspaces of $L_1$ (for all eigenvalues). The relation~\eqref{intertw1} can be taken as a definition of the Darboux transformations.  The intertwining relation~\eqref{intertw1} appeared, for Sturm--Liouville operators, in work  of Shabat~\cite[eq.~19]{shabat:1992},  Veselov--Shabat~\cite{veselov-shabat:dress1993}, and Bagrov--Samsonov~\cite{bagrov-samsonov:factorization1995}.
Intertwining relation~\eqref{intertw1} is also related with supersymmetric quantum mechanics initiated by E.~Witten~\cite{witten1982}, see in particular~\cite{Cooper_Khare95,Ioffe_Junker99}. In 2D case, intertwining relation~\eqref{intertw1} was studied in the series of papers by 
A.~Andrianov, F.~Cannata, M.~Ioffe, see e.g.~\cite{ioffe2010separation_of_vars} and references therein.
It also appeared for higher dimensions, in Berest--Veselov~\cite{berest-veselov:1998, berest-veselov:2000} 
for the Laplace type operators $L=-\Delta +u$. The same relation was used in~\cite{2015:super} for differential operators on the superline. 

\begin{displayquote}
\textit{{Problem 1: classification of all Darboux transformations defined by intertwining relation~\eqref{intertw1}.}}
\end{displayquote}

Problem 1 is solved for operators on the line and on the superline, which is an analogue of the 1D case, but 
has in principle two variables, one even (``bosonic'') and one odd (``fermionic'').
In both of these cases, it was established that 
all Darboux transformations (for any operator $L$; in supercase an extra condition of non-degeneracy is required) is a composition of atomic Darboux transformations of first order which arise from seed solutions and are given by ``Wronskian formulas''. In the supercase, 
these involve Berezinians. 

In the classical setting for the Sturm--Liouville operators, this was   proved in~\cite[Theorem 5]{veselov-shabat:dress1993} when the new potential $u_1(x)$  differs
from the initial  $u(x)$ by a constant, $u_1(x)=u(x) + c$.
It was proved in~\cite{shabat1995} for  transformations of order two; and, finally in the general case, 
in~\cite{bagrov-samsonov:factorization1995} and the follow-up paper~\cite{clouds:bagrov:samsonov:97}, 
see also~\cite[Sec.3]{samsonov:factorization1999}. Finally, for general operators on the line, 
it was proved in~\cite{adler-marikhin-shabat:2001}. 
For the superline, the classification was obtained in~\cite{2015:super,shemya:voronov2016:berezinians}.

The intertwining relation~\eqref{intertw1}, with a single transformation operator $M$, 
implies a linear mapping between the eigenspaces of the operators $L$ and $L_1$. It turns out that such symmetries are too restrictive 
for differential operators in higher dimensions. So the following more flexible version is considered,
which can be extracted  from the work of Darboux himself~\cite{Darboux2} 
(see for example~\cite[eq.~2]{TsarevS2009}):
\begin{equation}\label{intertw2}
NL=L_1M\,.
\end{equation}
Relation~\eqref{intertw2} implies a linear mapping of the kernels only. 
The connection between intertwining relations~\eqref{intertw1} and~\eqref{intertw2} is yet to be understood. 
In the present paper we solve the following problem for an important class of operators. 

\begin{displayquote}
\textit{{Problem 2: classification of all Darboux transformations defined by intertwining relation~\eqref{intertw2}.}}
\end{displayquote}
 

Note that one of the differences between the two intertwining relations, is that the 
intertwining relation~\eqref{intertw2} admits the following equivalence relation: $(M,N) \sim (M+AL,N+L_1A)$, for an arbitrary linear partial differential operator $A$.
The importance of this is that using it one can define invertible Darboux transformations in a natural
way, first suggested in~\cite{2013:invertible:darboux}, 
see also in the introduction of~\cite{shemya:hobby2016:iterated}  (and more in Sec.~\ref{sec:Known_types_of_Darboux transformations}). It it also the point of view that we have to take if we want
to factor Darboux transformations, and this is the one we use in this paper. 
Putting everything together, we essentially have a category where objects are operators $L$ of the same fixed principal symbol 
and Darboux transformations are morphisms in this category (again see more in Sec.~\ref{sec:Known_types_of_Darboux transformations}).
 
In work~\cite{shem:Darboux2}, we  classified Darboux transformations
for operators of the form
\begin{equation} \label{op:L2}
{L} = \partial_x\partial_y +a \partial_x + b\partial_y +c 
\end{equation}
with $a=a(x,y)$, $b=b(x,y)$, $c=c(x,y)$, where operator $M$ is  
of order $2$. This type of operators~\eqref{op:L2} can be called `2D Schr\"odinger operator' in `algebraic form'~\footnote{Calling
it Schr\"odinger is some abuse of language justified by the fact that on the formal algebraic level it is
equivalent to the actual elliptic Schr\"odinger operator~\cite{novikov-dynnikov:1997}.},
In~\cite{shem:Darboux2}, regularized moving frames of 
Fels and Olver~\cite{OP:09} were used to invariantize the system of nonlinear partial differential equations
describing the Darboux transformations. This system, when transformed on the space of joint differential invariants of the 
operators $M$ and $L$, becomes simple and can be solved exactly. 

In the present paper, we prove factorization for Darboux transformations of arbitrary order defined by the intertwining relation~\eqref{intertw2},
 for operators of type~\eqref{op:L2}. Specifically we proved that every Darboux transformation can be presented as a consequent application of Darboux transformations of two classically known types:  Wronskian type and Laplace transformations. Such a factorization is only possible if we consider Darboux transformations as equivalence classes up to the equivalence relation mentioned above (see Definition~\ref{def:darb}). We obtain our result under the technical assumption of the existence of a long enough Laplace chains for the original operator. 
  
 The proof is an induction on the order of the operator $M$. At every step,
if the intersection of the kernels of $L$ and $M$ is non-zero, we use a common element $\varphi$ to construct the first-order operator
$M_\varphi=\partial_x - \varphi_x \varphi^{-1}$, and $M=M' M_\varphi$. 
We then show that for $M'$ and $M_\varphi$ there are corresponding Darboux transformations which provide a factorization of the initial Darboux transformation.
If the intersection of the kernels of $L$ and $M$ is zero, we show that we can similarly factor out 
a Laplace Darboux transformation.

Since the first versions of this work appeared on the arXiv, its results 
were commented in several papers. Very recently, in~\cite{Smirnov2018},
S.~Smirnov by a modification of the proof in this paper was able to establish the discrete and semi-discrete analogues of our main statement. 

The structure of the paper is simple. In Sec.~\ref{sec:Known_types_of_Darboux transformations}, 
we recall the key results about Darboux transformations, their known types, and as much as it is known so far about 
their general structure. In Sec.~\ref{sec:fac} we give a proof of the main statement.

%
\section{Structure and known types of Darboux transformations}
\label{sec:Known_types_of_Darboux transformations}

In this section we recall some general facts about Darboux transformations. 
Consider a differential field $K$ of characteristic zero with a number of commuting derivations.
We shall work with the corresponding ring of linear partial differential operators over $K$.
One can either assume the field $K$ to be differentially closed, or simply assume that $K$ contains the solutions of those partial differential equations that we encounter on the way.

Darboux transformations viewed as mappings of linear partial differential operators can be defined algebraically as follows,
where we write $\sigma(L)$ for the principal symbol of $L$. 
\begin{definition}\label{def:darb}~\cite{shem:Darboux2,2015:inv:charts}
We introduce a category where
the objects are linear partial differential operators with the same arbitrary but fixed principal symbol, and the morphisms are defined as follows. 
A pair $({M},{N})$ defines a \emph{morphism} from ${L}$ to ${L}_1$ if 
\[
{N} {L} = {L}_1 {M}
\]
and such pairs $({M},{N})$ are considered up to the equivalence
\[
({M},{N}) \sim ({M} +{A} {L} \  , 
{N} + {L}_1 {A} ) \,,
\]
where ${A}$ is an arbitrary linear partial differential operator. Another notation for this morphism indicating the source and target, is $\t{L}{(M,N)}{L_1}$. In the notation, we shall not distinguish between a pair $(M,N)$ and its equivalence class (the corresponding morphism).
Morphisms with compatible source and target are composed as
%
$(M,N)\cdot(M_1,N_1) =(M_1 M ,N_1 N)$. 
For any object $L$, \textit{the identity morphism} is $1_L= (1,1)$.
So defined morphisms are called \textit{Darboux transformations}. 
\end{definition}

Every Darboux transformation $(M,N)$ defines a linear mapping of the  kernels, from $\ker L$ to $\ker L_1$, where a function $\phi$ is mapped to $M \phi$. The equivalence relation above then have a natural meaning:
since $(M+AL)\phi = M\phi$, the operators  $M$ and $M+AL$ give the same linear map on the kernel of $L$.

Note if $L$ and $M$ are known, then the coefficients of $L_1$ and $N$ are defined from a system of linear algebraic equations. Also not every operator $M$ can give a Darboux transformation for given operator $L$. As we shall see, intertwining relation~\eqref{intertw2} is, in a sense, an overdetermined system and $M$ must satisfy compatibility conditions.


In~\cite{shem:Darboux2} it was proved that the composition of morphisms is well defined (one has to check the compatibility with the equivalence relation).

Observe that given $NL = L_1M$, we have $\sigma(L) = \sigma(L_1)$ iff $\sigma(N) = \sigma(M)$.  
Note also that we have defined a category for the operators with fixed principal
symbol. 


We define the {\em order} of a Darboux transformation 
$(M,N)$ as the minimal possible order of a transformation in its equivalence class, that is the least possible order of an operator of the form $M+AL$. 

We also have the invertibility in terms
of equivalence classes, and say that the Darboux transformation $(M,N): L \rightarrow L_1$ is {\em invertible}
iff there exists $(M',N'): L_1 \rightarrow L$
such that the compositions 
$(M,N) \circ (M',N')$ and 
$(M',N') \circ (M,N)$ 
are identity morphisms (i.e. equivalent to identity morphism in the category).
By expanding the definition we obtain that
$(M,N) \colon L \rightarrow L_1$ and 
$(M',N') \colon L_1 \rightarrow L$
are mutual inverses if and only if there are
some operators $A,B$ so that the following hold:
\begin{align}
&{{M'M=1+AL}} \ , \label{inv:prop1} \\
&N'N=1+LA \ , \label{inv:nl} \\
&{{MM'=1+BL_1}} \ , \label{inv:prop2} \\ 
&NN'=1+L_1B \ .  \label{inv:nl1}
\end{align}

\begin{lemma} Invertible Darboux transformations induce isomorphisms on the kernels of the operators $L$ and $L_1$.
\end{lemma}

In particular, \eqref{inv:prop1} implies that $\ker {L} \, \cap \, \ker {M}=\{0\}$  is necessary for a Darboux transformation to be invertible.  Also note that the order of a Darboux transformation is not related in an obvious way to the order of its inverse.

The first steps in the general study of invertible Darboux transformations along with the analysis of one particular invertible class of Darboux transformations\,---\,the Darboux transformations of Type I\,---\,can be found in~\cite{2013:invertible:darboux,2015:inv:charts}.
 
\begin{lemma}  \label{lem:gn=ma} 
\begin{enumerate}
    \item Given the intertwining relations $NL = L_1 M$ and $N' L_1 = LM'$, the equalities for $N, N'$ follow from the equalities for $M, M'$.  Specifically, \eqref{inv:nl} follows from \eqref{inv:prop1}, and \eqref{inv:nl1} follows from \eqref{inv:prop2},
    \item If $(M,N)$ and $(M',N')$ are mutually inverse morphisms satisfying \eqref{inv:prop1} through \eqref{inv:nl1}, then $BN=MA$.
\end{enumerate}
\end{lemma}

Recall that a \textit{gauge transformation} of an operator $L$ is defined 
as $L^{g}= g^{-1} L g$, where $g \in \Ko$. Gauge transformations commute with Darboux transformations in the sense that if there is a Darboux transformation $(M,N): L \rightarrow L_1$, then there are also Darboux transformations 
\begin{align*}
    &(M^g,N^g): L^g \rightarrow L_1^g  \ , \\
    &(Mg,Ng): L^g \rightarrow L_1  \ , \\
    &(g^{-1}M,g^{-1}N): L \rightarrow L_1^g  \ .
\end{align*}

\begin{definition}\label{shift definition}
If a pair $(M,N)$ defines a Darboux transformation from $L$ to $L_1$, then the same 
pair defines a different Darboux transformation from $L+AM$ to $L_1+NA$ for every operator $A$.
We shall say that Darboux transformations $L \rightarrow L_1$ and $L+AM \rightarrow L_1+NA$ are related by a {\em shift}. Note that this is a symmetric relation.
\end{definition}

The shift of the Darboux transformation 
$(M,N) \colon L \rightarrow L_1$
is a Darboux transformation.  We have that 
$N(L+AM) = (L_1 + NA)M$ for the intertwining relation,
and $\sigma(L) = \sigma(L_1)$ implies
$\sigma(M) = \sigma(N)$, which implies $\sigma(L+AM) = \sigma(L_1 + NA)$.

In general, shifts of Darboux transformations do not commute with compositions of Darboux transformations. However, shifts are useful when dealing with invertible Darboux transformations. The following lemma was proved in~\cite{shem:Darboux2}.

\begin{lemma}\label{shift invertibility lemma}
Shifts of Darboux transformations preserve invertibility.
\end{lemma}

Darboux transformations and their variations are actively studied by physicists and mathematicians,
and there is a great variation in terminology and understanding of the term. E.g. one author's
``Laplace transformation'' can be called ``Darboux transformation'' by another author. Authors motivated by Physics
call Darboux transformations what we below name ``Wronskian type''. Ultimately all 
structures are interrelated and useful. Below we list several types of such transformations. All of them except 
Ganzha's Intertwining Laplace transformations are Darboux transformations as defined above.

\smallskip
\newcounter{Darboux transformationscount}
\stepcounter{Darboux transformationscount}
\noindent 
\textbf{\arabic{Darboux transformationscount}\quad\textit{Darboux transformations of Wronskian Type}}\\

These are the classical Darboux transformations. The operator $M$ is given by 
\[
M(f) = \frac{W(f_1, f_2, \dots , f_m,f)}{W(f_1, f_2, \dots , f_m)}
\]
where $W(f_1, f_2, \dots , f_m)$ denotes a Wronskian determinant with respect to one of the variables $t,\vars$ of $m$ linearly independent $f_j \in \Ko$,
which are 
elements of $\ker L$. 

For arbitrary operators on the line, it was proved by V. Matveev~\cite{Matveev79} that 
a composition of consecutive application of  $n$ Darboux transformations of this type of order $1$ (he calls them $1$-fold Darboux transformations) is given by a Wronskian formula of the corresponding 
larger size (more exactly, a  Wronskian of size $n+1$ in the numerator, and of size $n$ in the denominator; he calls them $n$-fold Darboux transformations).

Wronskian Type Darboux transformations are proved to be admitted by several different types of operators~\cite{Matveev79,matveev:1991:darboux,nimmo2010,2015:super}, and direct applications of these transformations for solving famous nonlinear equations are well known. Wronskian
type transformations are never invertible, since $\ker L \cap \ker M \neq \{0\}$. 
(An abstract framework for Wronskian type transformations was considered by Etingof-Gelfand-Retakh~\cite{etingof-gelfand-retakh:1997} and  Li-Nimmo~\cite{nimmo2010}. An analog of Wronskian Darboux transformations for super Sturm--Liouvlle operators was discovered in Liu and Liu-Ma\~nas~\cite{liu95_initial,liu_manas97_2,liu_manas97}.)

Note that in 1D case with intertwining relation having $M=N$, \eqref{intertw1}, similar Wronskian construction is used (same formula), but $f_i$
are arbitrary eigenfunctions with not necessarily zero eigenvalues. For the 1D case, every Darboux transformations is of Wronskian type in this extended sense, 
see in~\cite{adler-marikhin-shabat:2001}.
There is a recent very general result describing all operators admitting first-order Darboux transformations of Wronskian type~\cite{shemya:hobby2016:iterated}.

\begin{definition}
For a given variable $t$, call a differential operator 
{\em $t$-free} if it does not contain $\dee_t$
and none of its coefficients depend on $t$.
\end{definition}

\begin{theorem}[a criterion] 
\label{thm:wronski_type_char}
Given a linear partial differential operator $L$ and a $\psi \in \ker L$, the operator $M = \partial_t-\psi_t \psi^{-1}$ generates a Darboux
transformation if and only if there exists some operator $A$, and a $t$-free operator $B$, and non-zero function $c$ possibly depending on $t$ such that, for the gauge transformed operator,   
\[
L^{\psi} = A \partial_t + c B \ .
\]
Here either $B = 0$, or any one non-zero coefficient in $B$ may be taken to be $1$. If the DT exists, then $L_1$ and $N$ are given by 
$L_1^{\psi} =  L^{\psi} - A \partial_t +N^{\psi}A$, $N^\psi     =  \partial_t^{1/{c}}$, where $\dee_t^{1/c}=c \dee_t\circ  (1/c)$.
\end{theorem}

\begin{corollary}[a necessary condition]
\label{cor:1}
Suppose for some operator $L$, $\psi \in \ker L$, $M_\psi = \partial_t-\psi_t \psi^{-1}$ generates a Darboux transformation. Then there exists some operator $A$, and a $t$-free operator $B$ such that  \[
L = A M_\psi + c B \ , \quad [M_\psi,B]=0 \ .
\]
Here $B$ is not necessarily $t$-free.
\end{corollary}

\smallskip
\stepcounter{Darboux transformationscount}
\noindent 
\textbf{\arabic{Darboux transformationscount}\quad
\textit{Darboux transformations obtained from a factorization}}\\ 
Suppose $L=CM$ for some $C, M \in \K \setminus K$. 
Then for any operator $N$ with 
$\sigma(N) = \sigma(M)$, there is a Darboux transformation
$$(M,N) \colon CM \rightarrow NC \  , $$ since
$N(CM) = (NC)M$.  Since $\ker M \subseteq \ker L$,
these transformations are never invertible.

In the proofs, one can frequently see a trick where the  transformation operator $M$ can be (using some reasoning)
considered in a form where it is effectively an ordinary differential operator while everything else is multidimensional. In this
case if Darboux transformation is of factorization type (and with this particular $M$), then it is also of Wronskian type. This uses
the fact that every linear ordinary differential operator can be expressed by a Wronskian formula.

\smallskip
\stepcounter{Darboux transformationscount}
\noindent 
\textbf{\arabic{Darboux transformationscount}\quad
\textit{Laplace transformations}}\\
These are another type of Darboux transformations,
introduced in ~\cite{Darboux2}.
They are distinct from the Wronskian Type. They are only defined for 2D Schr\"odinger type operators, which have the form 
\begin{equation} \label{eq:L}
L = \partial_x\partial_y +a \partial_x + b\partial_y +c \, ,    
\end{equation}
where $a, b, c \in K$.
By interesting coincidence, the obstructions $h$ and $k$ to factorization of this operator are gauge (differential) invariants of operators of this form:
\begin{align}
  \o{L} &= (\partial_y + a)(\partial_x + b) - k \ , \label{eq:incomplet_fac_k} \\
        &= (\partial_x + b)(\partial_y + a) - h \ . \label{eq:incomplet_fac_h}
\end{align}

If the {\em Laplace invariant} $k = b_y + ab -c$ is nonzero, then $L$ admits a Darboux transformation with $M = \partial_x + b$ (in the  ``$x$-direction'').
(Explicit formulas for $L_1$ and $N$ are given below.) 
If the other Laplace invariant, $h = a_x + ab -c$
is nonzero, then $L$ admits a Darboux transformation with $M = \partial_y + a$ (in the  ``$y$-direction'').

Laplace transformations are invertible, and are (almost) inverses of each other. 
This has been mentioned in the literature, e.g.~\cite{ts:steklov_etc:00}. Classically, the invertibility of Laplace transformations was understood to mean that they induce isomorphisms of the kernels of the operators in question.
\comm{Interestingly, it is exactly the equivalence relation on Darboux transformations} that makes it possible to understand the invertibility of Laplace transformations in the precise algebraic sense. The detailed proof of the following (classical) statement can be found in~\cite{2013:invertible:darboux,shemya:hobby2016:iterated}.

\begin{theorem} 
\begin{enumerate}
    \item The composition of two consecutive Laplace transformations applied to $L$, first in $x$ direction, and then in $y$ direction is equal to the gauge transformation $L \rightarrow L^{1/k}$. 
If the transformation is first in $x$ direction, and then in $y$ direction then the composition is equal to the gauge transformation $L \rightarrow L^{1/h}$.
    \item 
The inverse for the Laplace transformation $L \rightarrow L_1$ given by the operator $M=\partial_x + b$ is $(M',N'): L_1 \rightarrow L$, where
$M'=-k^{-1} \left( \partial_y + a\right)$,  $N'=-k^{-1} \left( \partial_y + a- k_y k^{-1} \right)$. 
The inverse for the Laplace transformation $L \rightarrow L_1$ given by the operator $M=\partial_y + a$ is $(M',N'): L_1 \rightarrow L$, where
$M'=-h^{-1} \left( \partial_x + b\right)$,  $N'=-h^{-1} \left( \partial_x + b- h_y h^{-1} \right)$. 
\end{enumerate}
\end{theorem}

(Note that the formulas for the inverse for Laplace transformations can be generalized, as we do below, for Darboux transformations of Type I, which is a generalization of Laplace transformations to operators of a more general form.)

\smallskip
\stepcounter{Darboux transformationscount}
\noindent 
\textbf{\arabic{Darboux transformationscount}\quad
\textit{Ganzha's Intertwining Laplace transformations}}\\
These were introduced in~\cite{ganzha2013intertwining},
and generalize Laplace transformations to linear partial differential operators $L \in \K$ of very general form. 
One starts with any representation $L=X_1X_2-H$, where $L, X_1, X_2, H \in \K$. Then it was proved that there is a Darboux transformation for the operator $L$,
\[
(X_2,X_2+\omega): \ L \rightarrow L_1 \ ,
\]
where $L_1=X_2X_1+\omega X_1 - H$, and $\omega= -[X_2,H] H^{-1}$. The latter is a pseudodifferential operator
in the general case, an element of the skew Ore field over $\Ko$ that extends $\K$. 
In ~\cite{ganzha2013intertwining}, E. Ganzha then adds
the requirement that $\omega$ is a differential operator.
This is a very general class of transformations and contains both invertible and non-invertible Darboux transformations.

\smallskip
\stepcounter{Darboux transformationscount}
\noindent 
\textbf{\arabic{Darboux transformationscount}\quad 
\textit{Darboux transformations of Type I}}\\
These were introduced in~\cite{2015:inv:charts}, and  
are admitted by operators in $\K$ that can be written in the form $L=CM+f$, where $C,M \in \K$, and $f \in \Ko$. We have 
\[
(M, M^{1/f}): \ L \rightarrow L_1 \ ,
\]
where 
$L_1=M^{1/f}C+f$, writing $M^{1/f}$ for $f M (1/f)$. The inverse Darboux transformation always exists and is as follows:
\[
\left(-\frac{1}{f} \ C, - \ C \frac{1}{f}\right): \ L_1 \rightarrow L \ . 
\]
The original theory of Laplace transformations is very naturally
formulated in terms of differential invariants. In~\cite{2015:inv:charts}, 
analogous ideas were developed for Darboux transformations of Type I for operators of third order and in two independent variables. This can also be done
for operators of a general form using regularized moving frames~\cite{FO1,FO2} and ideas from~\cite{movingframes}.   
The classical Laplace transformations are a  special case of transformations of Type I. 
We will see that Darboux transformations of Type I can be identified with a subclass of the Intertwining Laplace transformations defined above. 

Note that the composition of two Darboux transformations of Type I is (in general) not of Type I.

It was recently proved~\cite{shemya:hobby2016:iterated}
that every first order Darboux transformation for every operator is either of Type I or of Wronskian type.

It is interesting to find out the exact relation between Intertwining Laplace transformations of Ganzha~\cite{ganzha2013intertwining} and  
Darboux transformations of Type I. 

\smallskip
\stepcounter{Darboux transformationscount}
\noindent 
\textbf{\arabic{Darboux transformationscount}\quad 
\textit{Continued Type I Darboux transformations}}

\noindent 
Continued Type I Darboux transformations were 
introduced recently in~\cite{shemya:hobby2016:iterated} as a further generalization of Darboux transformations of Type I.
They are present when $L$ can be divided by some $M$, and then $M$ divided by the ``remainder'', and so on, until we have a function.
All these transformations are invertible, and the explicit formulas for the inverses can be obtained by induction. 

\begin{definition}\label{iterated type definition}
Suppose that for operator $L$ and some operator $M=M_1$
there are nonzero
operators $A_1, A_2 \dots A_k, M_2, \dots M_k$ 
and a nonzero function $f = M_{k+1} \in K$
so that 
\begin{align*}
L  &= A_1 M_1 + M_{2} \\
M_1 & = A_2 M_2 + M_3 \\
 & \dots \\
M_{i-1} &= A_i M_i + M_{i+1} \ \mbox{, for} \quad 1 \leq i \leq k \ \mbox{, and finally} \\
& \dots \\
M_{k-1} & = A_k M_k + f 
\end{align*}

Then there exists a Darboux transformation for operator $L$
\[
(M,N): L \rightarrow L_1
\]
where $N_{k+1} = M_{k+1} = f$,
$N_{k} = fM_kf^{-1}$, $N = N_1$, $L_1 = N_0$ and  
$N_i$ for $0 \leq i \leq k-1$ can be found by
downward recursion using
\begin{equation*}
N_i = N_{i+1}A_{i+1} + N_{i+2} \ .
\end{equation*}
The corresponding Darboux transformation
is not obtained from a factorization, nor a multiple of a DT of Wronskian Type, and if $k > 1$ it is not of Type I.
We shall call these Darboux transformations {\em Continued Type I Darboux transformations}.
\end{definition}

%

\smallskip
\stepcounter{Darboux transformationscount}
\noindent 
\textbf{\arabic{Darboux transformationscount}\quad 
\textit{Athorne's generalized Laplace transformations}}

\noindent 

Ch.~Athorne~\cite{Athorne2018} has started recently an entirely original approach towards generalization of Darboux transformations (which he prefers to call Laplace transformations) for a particular but large class of operators of the form: 
\[
L= 
\partial_1\partial_2\partial_3+a_1\partial_2\partial_3+a_2\partial_1\partial_3+a_3\partial_1\partial_2+a_{12}\partial_3+a_{23}\partial_1+a_{13}\partial_2+a_{123} \ .
\]
The approach has potential to be generalized. A particular restriction is that the order of the operator is equal to the number of independent variables. 
A key role is played by differential invariants (relative to the gauge transformations of the operator).


Note also that in the Physics literature there are examples of Darboux transformations for concrete differential operators such as non-stationary Schr\"odinger 
operator or Fokker-Planck operator, see~\cite{Ioffe_Junker99}. 

Also analogues of Darboux transformations are used for discrete and semi-discrete cases, and also for operators on a space scale, see recent works of S.~Smirnov~\cite{Smirnov2018}, and G.~Hovhannisyan et al.~\cite{HOVHANNISYAN2018776, HOVHANNISYAN20161690}.



\section{The factorization problem for $L= \partial_x \partial_y +a \partial_x + b\partial_y +c$}
\label{sec:fac}
\subsection{Factorization of Darboux transformations}

Suppose there is a Darboux transformation    $L\to L_1$ defined by $(M,N)$. When can it be factorized into a composition of Darboux transformations $L\to L_0$ and $L_0\to L_1$?
Theoretically speaking, to obtain a factorization 
\begin{equation*}
    L\stackrel{(M_0,N_0)}{\longrightarrow}{L_0}\stackrel{(M',N')}{\longrightarrow}{L_1}
\end{equation*}
we need: \\
(1) $M=M'M_0$ (a factorization of $M$); \\
(2) $N=N'N_0$ (a factorization of $N$); \\
(3) so that $(M_0, N_0)$ defines a Darboux transformation $L\to L_0$; \\
(4) so that $(M', N')$ defines a Darboux transformation $L_0\to L_1$.

First of all notice a simple but useful observation: to establish a factorization of a Darboux transformation, it is enough to 
establish (1), (2), (3). 

\begin{theorem}\label{thm:div_darboux}
Statements (1), (2), (3) above imply (4).
\end{theorem}
\begin{proof}
The equality ${N} {L} = {L}_1 {M}$ can be re-written as follows ${N}' {N}_0 {L} = {L}_1 {M}' {M}_0$, which implies
${N}' {L}_0 {M}_0 = {L}_1 {M}' {M}_0$, and thus, $\left( {N}' {L}_0 - {L}_1 {M}' \right) {M}_0  = 0$.
Then the consideration of the principal symbols implies that ${N}' {L}_0 - {L}_1 {M}'=0$.
\end{proof}


\subsection{Reduction to the case where ${M}$ is an ordinary differential operator}
\label{sec:make_ordinary}

The special form of  the  operators $L= \partial_x \partial_y +a \partial_x + b\partial_y +c$ implies some special properties of Darboux transformations for them. Firstly, it makes it possible to find an operator ${A}$ so that ${M} +{A} {L}$
does not contain mixed derivatives. That means  that in every equivalence class we can choose
a \textit{standard representative} --- a pair $({M},{N})$
such that ${M}$ does not contain mixed derivatives. We   call it the \textit{projection of ${M}$ with respect to ${L}$}
and denote $\pi_{{L}}({M})$. Thus, for some $\alpha_i$, $\beta_j$, $m_0$
\begin{equation} \label{eq:M_no_mixed_derivatives}
\pi_{{L}}({M}) = \alpha_i \partial_x^i + \beta_j \partial_y^j + m_0 \  ,
\end{equation}
where we assume summation over the indices $i=1, \dots, d_1$ and   $j=1, \dots, d_2$.

Let ${M}$ be an operator of the form~\eqref{eq:M_no_mixed_derivatives}, where $\alpha_{d_1} \neq 0$ and
$\beta_{d_2} \neq 0$. Then we define the \textit{bi-degree} of ${M}$ as
\[
\deg({M})=(d_1,d_2) \ .
\]

If ${M}$ is an arbitrary operator in $\K$, then we define the \emph{bi-degree} of such operator with respect to
some given operator ${L}$ of the form~\eqref{op:L2} as the bi-degree of its projection
\[
\deg_{{L}}({M})=\deg(\pi_{{L}}({M})) \ .
\]

Consider a Darboux transformation of a 2D Schr\"odinger operator with $M$ of the form~\eqref{eq:M_no_mixed_derivatives}. It can be further reduced to an ordinary differential operator of order less than $d_1+d_2$.

\begin{theorem} \label{thm:compose_with_laplace}
Let ${L}$ be an operator of the form~\eqref{op:L2} and
${M}$ be an operator of the form~\eqref{eq:M_no_mixed_derivatives} and $\deg {M} =(d_1,d_2)$.
Then for operators ${M}_y = \partial_x + b$, ${M}_x = \partial_y + a$ defining two Laplace transformations, the following is true.

(1) $\deg_{{L}}({M} {M}_y) = (d_1+1,d_2-1)$, where $d_2 \neq 0$.

(2) $\deg_{{L}}( {M} {M}_x) = (d_1-1,d_2+1)$, where $d_1 \neq 0$,

(3) $\pi_{{L}}({M}_x {M}_y) = k$ ($d_2=0$ in this case), 

(4) $\pi_{{L}}({M}_y {M}_x) = h$ ($d_1=0$ in this case), 

\noindent where $h$ and $k$ are the Laplace invariants.
\end{theorem}
\begin{proof}
(1). Let us consider ${M}$ in the form ${M} = {M}^x + {M}^y + m_0$, where
$ {M}^x = \sum_{i=1}^{d_1} \alpha_i \partial_x^i$,  ${M}^y = \sum_{j=1}^{d_2} \beta_j \partial_y^j $,
and then consider first ${M}^y (\partial_x+b) - {A} {L}$ for some ${A} \in \K$. Using representation~\eqref{eq:incomplet_fac_k} of ${L}$, we see that
$ {M}^y (\partial_x+b) - {A} {L} =  {M}^y (\partial_x+b) - {A}  (\partial_y+a)(\partial_x+b) + {A} k
=  \left( {M}^y - {A}  (\partial_y+a) \right) (\partial_x+b) + {A} k$.
Choose ${A}$ so that ${M}^y = {A}  (\partial_y+a) + r$ for some $r \in K$. In this case ${A}$ is an ordinary differential operator in $\partial_y$ and
$\deg {A} = d_2-1$. Hence, ${M}^y (\partial_x+b) - {A} {L} =  r (\partial_x+b) + {A} k$, and
$  {M} (\partial_x+b) - {A} {L} = {M}^x (\partial_x+b) + m_0 (\partial_x+b) + r (\partial_x+b) + {A} k
= \left({M}^x + m_0+r \right)(\partial_x+b) + {A} k$.
As we use the Laplace transformation defined ${M}_y$, Laplace invariant $k$ is not zero, and thus the bi-degree of this operator is $(d_1+1,d_2-1)$.
Part (2) is similar.

(3). As $ {M}_x {M}_y - {A} {L} = (\partial_y+a)(\partial_x+b) - {A} {L} = (\partial_y+a)(\partial_x+b) - {A}  (\partial_y+a)(\partial_x+b) + {A} k
=  (1 - {A})(\partial_y+a)(\partial_x+b) + {A} k$, then if ${A}=1$, then ${M}_x {M}_y - {A} {L} = k$. Part (4) is similar.
\end{proof}

Applying one or several times Theorem~\ref{thm:compose_with_laplace}
we can completely eliminate either all the derivatives with $\partial_y$
or all the derivatives with $\partial_x$ in ${M}$, provided we can move far enough
along the chain of the Laplace invariants for the initial operator ${L}$. This can be formulated as
follows.

\begin{theorem} \label{thm:normalization_of_M}
Consider a Darboux transformation ${L} \xrightarrow{{M}} {L}_1$, where
${M}$ is an operator of the form~\eqref{eq:M_no_mixed_derivatives} and $\deg {M} =(d_1,d_2)$.
Then if the Laplace chain has length at least $d_2$ on the right, or
at least $d_1$ on the left, then there exists a sequence of Laplace transformations
\[
{{\widetilde{L}} \xrightarrow{{\widetilde{M}_1}} \dots \xrightarrow{{\widetilde{M}_t}} {L}} \ ,
\]
for some ${\widetilde{L}}$, so that there exists a Darboux transformation
${{\widetilde{L}} \xrightarrow{{\widetilde{M}}} {L}_1}$
with $\widetilde{M}$ depending on $\partial_x^i$ or $\partial_y^i$ only, and the diagram
\begin{displaymath}
    \xymatrix{
      {\widetilde{L}} \ar[rr]^{{\widetilde{M}_t} \dots {\widetilde{M}_1}}
                        \ar[dr]_{{{\widetilde{M}}} }
                                                                                  &         & {L}
                        \ar[dl]^{{M}}
                                                                                     \\
                                                                                  & {L}_1 &        \\
   }
\end{displaymath}
is commutative.
\end{theorem}

\subsection{The case of factoring out a Wronskian type Darboux transformation}
\label{sec:yes}

In this section we analyze the case $\ker {M} \cap \ker {L} \neq \{0\}$. We show that in this case the Darboux transformation has a Wronskian type factor.


Recall first the following statement, which follows from the fact that ordinary differential operators over $K$ form a
Euclidean ring.

\begin{theorem}  \label{thm212} Let ${M}$ be an ordinary differential operator in $\partial_x$ over $K$,
and $\psi \in \ker {M}  \setminus \{0\}$. Then
${M}={M}_1 \cdot \left( \partial_x - \psi_{x} \psi^{-1} \right)$
for some ${M}_1 \in K[\partial_x]$.
\end{theorem}

We recall here the following well known statement which is going to be used later.
\begin{proposition}\label{thm_mult_f}
Let ${L} = \partial_x\partial_y +a \partial_x + b\partial_y +c$, where $a,b,c \in K$. Then for every $f(y) \in K$ depending only on $y$, ${L}(f(y) \psi) = f(y) {L}(\psi) + f'(y) {M}_y(\psi)$,
where ${M}_y=\partial_x+b$.
\end{proposition}

The next statement is probably new in a sense that it never appeared in a rigorous form with all the formulas, however the existence of ``Wronskian type'' Darboux transformations for operator of the form ${L} = \partial_x\partial_y +a \partial_x + b\partial_y +c$ is well known. The proof is rather technical and is given in the appendix.

\begin{theorem} \label{thm:psi} [Existence of Wronskian transformations for the second order hyperbolic operator]
Consider operator ${L} = \partial_x\partial_y +a \partial_x + b\partial_y +c$,  $a,b,c \in K$ and a non-zero element $\psi$ in its kernel.

Then for $L$ there is a (``Wronskian type'') Darboux transformation ${N}_\psi {L} = {L}_\psi {M}_\psi$ with 
\begin{equation} \label{eq:m_psi2}
 {M}_\psi =\partial_x - \psi_{x} \psi^{-1} \ ,
\end{equation}
and where ${N}_\psi$ is as follows. \\
\noindent Case I (general case). If $T=\psi_x  \psi^{-1}+b \neq 0$, then 
\begin{equation}
 N_\psi=\partial_x - \partial_x \ln (T \psi) 
\end{equation}
and the dressing formulas are 
\begin{equation}
L_\psi=\partial_x \partial_y + a_1 \partial_x + b_1 \partial_y + c_1
\end{equation}
where 
\begin{equation*}
a_1 = a \ , \quad  b_1=b- \partial_x \left(\ln T \right) \ , \quad c_1 =  a_x - a \partial_x \ln (T \psi ) -   T \partial_y (\ln \psi)
\end{equation*}

\noindent  Case II. If $T=\psi_x  \psi^{-1}+b \neq 0$, then 
\begin{equation}
 N_\psi=\partial_x - n(x,y) 
\end{equation}
(no constrains on $n(x,y)$) and operator $L$ can be factored 
as 
\begin{equation} \label{eq:fact_case_11}
L= \left(\partial_y + a \right) \left( \partial_x  + b \right) 
\end{equation}
and the dressing formulas are 
\begin{equation*}
a_1 = a \ , \quad  b_1=n \ , \quad c_1 =  na+ a_x \ .
\end{equation*}
\end{theorem}

%
%
%

We can now formulate and prove the first building block of our inductive proof.


\begin{theorem} \label{thm:main}
Consider operator ${L} = \partial_x\partial_y +a \partial_x + b\partial_y +c$ with Laplace invariant $k$ being non-zero and let there be a Darboux transformation 
\begin{equation} \label{eq:main:thm1}
  {N} {L} = {L}_1 {M} \ ,
\end{equation}
where 
${M}$ is an ordinary differential operator in $\partial_x$ of arbitrary order $d$ over $K$. Suppose the kernels of 
$L$ and $M$ has an non-trivial intersection:
\[
\psi \in \ker {L} \cap \ker {M} \ , \quad  \psi \neq 0 \ .
\]
Then this Darboux transformation can be represented as a consecutive application
of two Darboux transformations, the first one of which is generated
by ${M}_\psi =\partial_x - \psi_{x} \psi^{-1}$.
\end{theorem}
\begin{proof}
Applying gauge transformation, ${P} \to e^{-f} {P} e^f$ 
to all the operators in the~\eqref{eq:main:thm1}, we can reduce the problem to the case where the coefficient $b$ of the operator ${L}$ is zero.

For a nonzero $\psi \in \ker {L} \cap \ker {M}$, Theorems~\ref{thm212} and \ref{thm:psi} imply that there exist a Darboux transformation for ${L}$ generated by ${M}_\psi$. Let it be defined by equality ${N}_\psi {L} = {L}_\psi {M}_\psi$. Theorem~\ref{thm:div_darboux} implies that to prove the theorem it is enough to prove that ${N}$ is divisible by ${N}_{\psi}$ on the right, that is  ${N}={N}' {N}_{\psi}$ for some ${N}' \in \K$.

According to Theorem~\ref{thm:psi} two cases are possible, but in Case II invariant $k$ is zero, so we need to consider Case I only.
In this case $T=\psi_x  \psi^{-1}+b \neq 0$, and ${M}_\psi$ generates a unique Darboux transformation defined by ${N}_\psi = \partial_x - \partial_x \ln (T \psi)$. Recall that $b=0$ now, and therefore,
$T=\psi_x  \psi^{-1} \neq 0$, and so $T \psi = \psi_x$, and so
${N}_\psi = \partial_x - \partial_x \ln (\psi_x)$.

As it follows from Proposition~\ref{thm_mult_f}, to prove that $N$
is divisible by ${N}_\psi$ it is enough to prove that non-zero $\psi_x$ is in the kernel of $N$. 

We prove it using a small trick and Proposition~\ref{thm_mult_f}.
Indeed, consider the action of the both sides of operator equality~\eqref{eq:main:thm1} on $y \psi$. We have then ${N} {L} (y \psi) = {L}_1 {M} (y \psi)$ implying ${N} \left(y {L} (\psi) +  \psi_x \right) = {L}_1 \left( y {M} (\psi) \right)$,
and then ${N} \left(\psi_x \right) = 0$.
Applying Theorem~\ref{thm212} for ${N}$ and non-zero element of its kernel $\psi_x$,
we  have ${N}={N}' {N}_{\psi}$ for some ${N}' \in \K$.

\end{proof}

\subsection{The case of factoring out a Laplace transformation}
\label{sec:no}

Suppose now $\ker {M} \cap \ker {L} = \{0\}$.
Here we exclude from consideration   the case where ${L}$ possesses a factorization of the form ${L}=(\partial_y +a)(\partial_x +b)$.
In this case the corresponding Laplace transformation does not exist. Instead we can still factor out a Wronskian type Darboux transformation, as it was shown in~\cite{shem_2014_factorization_of_DT_for_factorizable}.

\begin{theorem} \label{thm:main2}
Let for an operator ${L} = \partial_x\partial_y +a \partial_x + b\partial_y +c$ its Laplace invariant $k$ be non-zero and let there be a Darboux transformation
$\t{{L}}{{M}}{{L}_1}$
generated by ${M}$  which is an ordinary differential operator in $\partial_x$ of arbitrary order $d$ over $K$,
and
\[
 \ker L \cap \ker M= \{ 0 \} \ .
\]
Then the Darboux transformation can be represented as a consecutive application of the Laplace transformations generated by
${M}_y=\partial_x+b$
and some other Darboux transformation of order one.
\end{theorem}
The rest of the section is devoted to a proof of this theorem.

Let the given Darboux transformation be defined by the equality
\begin{equation} \label{eq:main2}
 {N} \cdot {L} = {L}_1 \cdot {M} \  .
\end{equation}
The equality implies that $N$ is an ordinary differential operator,
${N}= n_i \partial_x^i$, for some  $n_i \in K$, $i=0, \dots, d$.
Since $k \neq 0$, then there exists the Laplace transformation generated by ${M}_y=\partial_x + b$:
\begin{equation} \label{eq:lapl}
  {N}_y {L} = {L}_{-1} {M}_y \  .
\end{equation}
{Our goal} is to factor out this transformation from the initial Darboux transformation.
By Theorem~\ref{thm:div_darboux} this is true if there are corresponding factorization of
${M}$ and ${N}$: ${M} = {M}' {M}_y$, and ${N} =  {N}' {N}_y$ for some ${M'}, {N'}$.
However, here the factorization for $N$ follows from the factorization of $M$. Indeed,  ${M} = {M}' {M}_y$ implies that there exists
a nonzero
\[
  \widetilde{\psi} \in \ker M_y \setminus \ker L \, .
\]
Equality~\eqref{eq:lapl} implies that ${L}(\widetilde{\psi}) \in \ker {N}_y$,
while~\eqref{eq:main2} implies ${L}(\widetilde{\psi}) \in \ker {N}$.
Therefore, by Theorem~\ref{thm212}, $N$ is divisible by $N_y$.
By Theorem~\ref{thm212}, to show ${M} = {M}' {M}_y$ is the same as to show that there is some non-zero element in
the intersection of the kernels of ${M}$ and ${M}_y$.

{In the rest of the proof} we assume the opposite,
\[
 \ker {M} \cap \ker {M}_y = \{ 0 \}
\]
and show that it leads us to a contradiction.

The main intertwining relation~\eqref{eq:main2} implies a mapping
\[
 L : \ \ker M \rightarrow \ker N \, .
\]
Moreover, the equality ${L}(f(y) \psi) - f(y) {L}( \psi) = f'(y){M}_y(\psi)$ (proved in Proposition~\ref{thm_mult_f})
implies a $C(y)$-linear (by construction) mapping
\begin{equation} \label{ker:y}
M_y : \ \ker {M} \rightarrow \ker {N} \  .
\end{equation}
With an abuse of notation we denote these mappings ${L}$ and ${M}_y$ by the same symbols as the corresponding operators.

\begin{lemma} \label{lem:map}
The map $M_y$ in \eqref{ker:y} is an isomorphism of vector spaces over $C(y)$.
\end{lemma}
\begin{proof}
%
Show that ${M}_y$ sends bases to bases. Let $\{e_i \}$, $i=1, \dots , d$, be a $C(y)$-basis for $\ker {M}$.
Suppose there is a non-trivial linear relation $f^i {M}_y(e_i)=0$, where $f_i \in C(y)$. Then by linearity,
$0=f^i {M}_y(e_i) = {M}_y \left(f^i e_i \right)$.
Since the kernel of the mapping is trivial, we have $f^i e_i= 0$, which contradicts $\{e_i \}$ being a basis.
Therefore, ${M}_y$ is a monomorphism, and since $\ker {M}$ and $\ker {N}$ are of the same dimension, we conclude that
it is an isomorphism.
\end{proof}

\begin{lemma} \label{lem:basis}
One can choose a $C(y)$-basis $\{ e_i \}$, $i=1, \dots, d$, for $\ker {M}$, so that $u_i={L}(e_i)$, $i=1, \dots ,d$, form a $C(y)$-basis for $\ker {N}$.
\end{lemma}
\begin{proof} We construct such $\{ e_i \}$, $i=1, \dots, d$, by induction.

Choose arbitrarily a non-zero $e_1 \in \ker {M}$. Since the kernel of the mapping ${L}$ is trivial, then $u_1={L}(e_i)$
is non-zero too, and the vector space $\langle u_1 \rangle$ over $C(y)$ is one-dimensional.
This is the basis of induction.

Assume that there are $k$, $k < d$, linearly independent over $C(y)$ elements $e_i$ in $\ker {M}$, such that $u_i={L}(e_i)$, $i=1, \dots ,d$, are linearly independent
over $C(y)$ in $\ker {N}$.

Choose arbitrarily a non-zero $e_{k+1}  \in \ker {M}$ that is linearly independent (over $C(y)$)
with $e_i$, $i=1, \dots, k$.
If $u_{k+1}={L}(e_{k+1})$ is linearly independent with all $u_i$, $i=1, \dots, k$ (over $C(y)$), the step of induction is proved.

Suppose now that $u_{k+1}={L}(e_{k+1})$ is not linearly independent with  $\{ u_i \}$, $i=1, \dots, k$ (over $C(y)$).
Consider then $\{ v_{i} \}$, where $v_i ={M}_y(e_{i})$, $i=1, \dots, k+1$. By Lemma~\ref{lem:map}
they are linearly independent and span a vector subspace of dimension $k+1$ in $\ker {N}$.
On the other hand, $\{u_{i} \}$, $i=1, \dots, k$, are linearly independent too,
and span a vector subspace of dimension $k$ in $\ker {N}$.
Thus, there is at least one $v_i$ that is linearly independent with  $\{ u_i \}$, $i=1, \dots, k$ (over $C(y)$).
Without loss of generality, we can assume that this $v_i$ is $v_1$.

Consider $\widetilde{e}_{k+1}= e_{k+1} + y e_{1}$. Using Proposition~\ref{thm_mult_f} we compute
 $\widetilde{u}_{k+1} =  {L}\left(\widetilde{e}_{k+1} \right) =   {L}\left(e_{k+1}\right) + y {L} (e_1) + y' {M}_y(e_1)  =    u_{k+1} + y u_{1} + v_{1}$.
 Thus, vector space $\langle u_1, \dots, u_k, \widetilde{u}_{k+1} \rangle$ over $C(y)$ is $(k+1)$-dimensional.

Since the dimensions of the vector spaces (over $C(y)$) $\ker {M}$ and $\ker {N}$ are the
same, this process can be continued until the desired basis is constructed.
\end{proof}

Using Lemma~\ref{lem:basis} we can choose $\{ e_i \} \ , i=1, \dots, d$, a basis over $C(y)$ for $\ker {M}$, so that
$\{ u_i \}$, where $u_i={L}(e_i)$ and $\{v_i\}$, where $v_i={M}_y(e_i)$, $i=1, \dots ,d$,
are two bases of $\ker {N}$.

We define a $C(y)$-linear operator ${A}$ by its action on its basis elements: ${A}(u_i) = v_i$, $i=1, \dots ,d$, that is
\begin{equation} \label{eq:dia}
    \xymatrix{
      \ker {M} \ni e_i\ar[rr]^{{L}}      \ar[drr]_{{{M_y}} }         &         & u_i \in \ker {N}  \ar[d]^{{A}}           \\
                                                                    & & v_i \in \ker {N} \ar[llu]       \\
   }
\end{equation}
and the diagram is commutative. Below we show that this cannot be true, and by this we arrive to a contradiction with our assumption
that $\ker {M} \cap \ker {M}_y = \{ 0 \}$.

If  $A$ is not degenerate, let $F = (f^1, \dots, f^d)$ be a row of $f^i \in {C}(y)$. By Proposition~\ref{thm_mult_f}
${L}(f^i e_i) =  f^i {L}(e_i) + (f^i)' {M}_y (e_i) =  f^i u_i + (f^i)' v_i =  f^i u_i + (f^i)' A(u_i) =  (F +  F'A) U$,
where $U$ is a column $(u_1, \dots, u_d)^t$.
Since $A$ is not degenerate, then the system of linear ordinary differential equations
$F'A = -F$ has a non zero solution $\widetilde{F}=(\widetilde{f}^1, \dots , \widetilde{f}^d)$.
Thus ${L}(\widetilde{f}^i e_i)=0$ and since $\{ e_i \}$ is a basis for $\ker {M}$, $\widetilde{f}^i e_i$ belongs to the intersection
of the kernels of ${L}$ and ${M}$, which is trivial by the condition of the theorem. Thus,
$\widetilde{f}^i e_i = 0$, and since this is a non-degenerate linear combination of the basis vectors, we come to a contradiction.

If the matrix $A$ is degenerate, then
for some non-zero $\widetilde{u} \in \ker {N}$,
$A \widetilde{u} =0$.
Consider
its pre-image, a non-zero $\widetilde{e}$ such that ${L}(\widetilde{e})=\widetilde{u}$.
Since diagram~\eqref{eq:dia} is commutative, ${M}_y \left( \widetilde{e} \right)  = 0$, which contradicts with the assumption that
the intersection of the kernels of ${M}$ and ${M}_y$ is trivial.

This concludes the proof of Theorem~\ref{thm:main2}.

\subsection{The general statement}
\label{sec:main}
Theorems~\ref{thm:main} and~\ref{thm:main2}
and the result of paper~\cite{shem_2014_factorization_of_DT_for_factorizable} describing a particular case when the operator ${L}$ is factorizable (see more details in Sec.~\ref{sec:no}),
imply together the following theorem.

\begin{theorem} \label{thm:main_very}
Consider operator
\begin{equation} \label{eq:opL2_6}
{L} = \partial_x\partial_y +a \partial_x + b\partial_y +c \  ,
\end{equation}
where $a,b,c \in K$.
Suppose a Darboux transformation of ${L}$ is defined by a pair $({M}, {N})$, where the operator ${M}$ is of bi-degree $(d_1,d_2)$.
If the Laplace chain of ${L}$ has length at least $d_2$ on the right, or
at least $d_1$ on the left, then the Darboux transformation can be represented as the composition of first-order elementary Darboux transformations.
\end{theorem}

\section*{Acknowledgements} This work was finished while supported by the National Science Foundation under Grant No.1708033.

\section*{Appendix} For the completeness of the exposition we give here the proof that due to its technical nature was skipped in the main part of the text.
Although the theorem may be known in general, it had never appeared in this complete form with all the formulas in the literature.

\begin{theorem} [Existence of Wronskian transformations for the Laplace operator]
Consider operator ${L} = \partial_x\partial_y +a \partial_x + b\partial_y +c$,  $a,b,c \in K$ and a non-zero element $\psi$ in its kernel.

Then for $L$ there is a (``Wronskian type'') Darboux transformation ${N}_\psi {L} = {L}_\psi {M}_\psi$ with 
\begin{equation} \label{eq:m_psi1}
 {M}_\psi =\partial_x - \psi_{x} \psi^{-1} \ ,
\end{equation}
and where ${N}_\psi$ is as follows. \\
\noindent Case I (general case). If $T=\psi_x  \psi^{-1}+b \neq 0$, then 
\begin{equation}
 N_\psi=\partial_x - \partial_x \ln (T \psi) 
\end{equation}
and the dressing formulas are 
\begin{equation}
L_\psi=\partial_x \partial_y + a_1 \partial_x + b_1 \partial_y + c_1
\end{equation}
where 
\begin{equation*}
a_1 = a \ , \quad  b_1=b- \partial_x \left(\ln T \right) \ , \quad c_1 =  a_x - a \partial_x \ln (T \psi ) -   T \partial_y (\ln \psi)
\end{equation*}

\noindent  Case II. If $T=\psi_x  \psi^{-1}+b \neq 0$, then 
\begin{equation}
 N_\psi=\partial_x - n 
\end{equation}
(no constrains on $n$) and operator $L$ can be factored 
as 
\begin{equation} \label{eq:fact_case_1}
L= \left(\partial_y + a \right) \left( \partial_x  + b \right)
\end{equation}
and the dressing formulas are 
\begin{equation*}
a_1 = a \ , \quad  b_1=n \ , \quad c_1 =  na+ a_x \ .
\end{equation*}
\end{theorem}

\begin{proof} Let $N=\partial_x+n$,  $L_1= \partial_x\partial_y +a_1 \partial_x + b_1 \partial_y +c_1$,  $n, a_1,b_1,c_1 \in K$.

From the fact that $\psi$ is in the kernel of $L$, obtain $c$ as
\begin{equation*}
  c = -\left(a \psi_x+b \psi_y+\psi_{xy} \right) \psi^{-1}
\end{equation*}
Then comparing the coefficients at $\partial_x^2$ and then at $\partial_x \partial_y$, and then at $\partial_x$ on the both sides of $NL=L_1M$ (and recalling that $\psi$ is not zero), we have
\begin{equation*}
   a_1 = a \ , \quad b_1=b+n+\left( \ln \psi \right)_x \ , \quad c_1=na+a_x - \psi_y \psi^{-1} (\psi_x+b \psi) \ .
\end{equation*}
Comparing coefficients at $\partial_y$ and at $1$ (free), we have two equivalent equations of the form
\begin{equation}
  (\psi_x+b \psi) n +  \partial_x (\psi_x+b \psi) =0
\end{equation}

Case I. Assuming $\psi_x+b \psi \neq 0$, we obtain
\begin{equation*}
  n = - \frac{(b \psi)_x + \psi_{xx}}{b \psi + \psi_x} 
\end{equation*}
Substituting into the formulas above we get 
\begin{equation*}
 b_1=b- \partial_x \left(\ln T \right) \ , \quad c_1 =  a_x - a \partial_x \ln (T \psi ) -   T \partial_y (\ln \psi)
\end{equation*}
where we have introduced notation
\begin{equation}
T=\psi_x  \psi^{-1}+b \ .
 \end{equation}
In this notation we have 
\begin{equation}
 n = - \partial_x \ln (T \psi) 
\end{equation}
and also the case splitting condition $\psi_x+b \psi \neq 0$ is equivalent 
to $T \neq 0$.

Case II. Assume $\psi_x+b \psi = 0$, we obtain
a rather interesting small case: 
\begin{equation}
 a_1 = a \ , \quad b_1 = b \ , \quad c_ 1=n a+a_x
\end{equation}
and no constrains on $n$. 
 
It might be immediately interesting for the reader to look at Laplace invariants
of $L$ and $L_1$ ($h$, $k$ and $h_1$, $k_1$, respectively):
\begin{equation*}
  h =  a_x + \partial_x \partial_y \ln (\psi) \ , \quad
  k = 0 \ , \quad
  h_1=0 \ , \quad k_1= n_y -a_x \ .
\end{equation*}
So in this case the original operator $L$ is factorizable as
\begin{equation}
  L = \left(\partial_y + a \right) \left( \partial_x  + b \right) \ .
\end{equation}
This particular small but interesting case is further described in~\cite{shem_2014_factorization_of_DT_for_factorizable} (see more details in Sec.~\ref{sec:no}).

\end{proof}

\section*{Acknowledgement}
This material is based upon work supported by the National Science Foundation under
grant No.1708033.

\bibliographystyle{plain}
%

\end{document}